\documentclass{IEEEtran}
\usepackage{nopageno}

\usepackage{amsmath}
\usepackage{amsfonts}
\usepackage{amssymb}
\usepackage{amsthm}
\usepackage{enumerate}
\usepackage{bbm}
\usepackage{hhline}
\usepackage{hyperref}
\usepackage{xcolor}

\newtheorem{theorem}{Theorem}[section]
\newtheorem{lemma}[theorem]{Lemma}

\newtheorem{corollary}[theorem]{Corollary}

\newcommand{\suchthat}{\;\ifnum\currentgrouptype=16 \middle\fi|\;}

\newcommand{\F}{\mathbb{F}}
\newcommand{\R}{\mathbb{R}}

\DeclareMathOperator{\ent}{H}

\newcommand{\Prob}{\mathbb{P}}
\newcommand{\E}{{\rm I\kern-.3em E}}
\newcommand{\Var}{\mathrm{Var}}

\newcommand{\rcov}{r_{\mathrm{cov}}}
\newcommand{\rpack}{r_{\mathrm{pack}}}

\title{Improved efficiency for  covering codes matching the sphere-covering bound}

\author{
\IEEEauthorblockN{
Aditya Potukuchi\IEEEauthorrefmark{1}\thanks{The work of Aditya Potukuchi was supported by Mario Szegedy\rq{}s NSF grant CCF-1514164.}
Yihan Zhang\IEEEauthorrefmark{2}\thanks{The work of Yihan Zhang was supported by the Research Grants Council (RGC) of Hong Kong under Project Numbers 14300617, 14304418 and 14301519.}
}\\
\IEEEauthorblockA{
\IEEEauthorrefmark{1} Dept. of Computer Science,  Rutgers University, \href{mailto:aditya.potukuchi@cs.rutgers.edu}{aditya.potukuchi@cs.rutgers.edu}\\
\IEEEauthorrefmark{2} Dept.\ of Information Engineering, The Chinese University of Hong Kong, \href{mailto:zy417@ie.cuhk.edu.hk}{zy417@ie.cuhk.edu.hk}
}
}

\begin{document}

\maketitle

\thispagestyle{empty} 

\title{Improved efficiency for covering codes matching the sphere-covering bound}

\maketitle

\begin{abstract}
A covering code is a subset $\mathcal{C} \subseteq \{0,1\}^n$ with the property that any $z \in \{0,1\}^n$ is close to some $c \in \mathcal{C}$ in Hamming distance. For every $\epsilon,\delta>0$, we show a construction of a family of codes with relative covering radius $\delta + \epsilon$ and rate $1 - \ent(\delta) $ with  block length at most $\exp(O((1/\epsilon) \log (1/\epsilon)))$ for every $\epsilon > 0$. This improves upon a folklore construction which only guaranteed codes of block length $\exp(1/\epsilon^2)$. The main idea behind this proof is to find a distribution on codes with relatively small support such that most of these codes have good covering properties.
\end{abstract}

\section{Introduction}
{For a subset $\mathcal{C} \subseteq \{0,1\}^n$, we define the \emph{covering radius} $\rcov(\mathcal{C})$ as the smallest $r$ such that every $z \in \{0,1\}^n$ is at most a distance $r$ from some $c \in \mathcal{C}$ in Hamming distance. Equivalently, we say that $\mathcal{C}$ is a \emph{covering code} (or simply, a code) of relative covering radius $\frac{r}{n}$.

Covering codes play  a central role in  rate distortion theory and source coding. Also,  as a combinatorial object, a covering code can be used as a net to approximate any point in the Hamming space and hence are also natural objects to study.} 

{\paragraph{Our main result} For every $\delta > 0$ we exhibit an ``explicit''\footnote{Our usage of the word ``explicit'' here, while technically correct in the sense that the description of the code is well defined and involves no probability distribution, is somewhat different from the standard usage of the word. We emphasize that our result is definitely of the ``with-high-probability'' mold. } family linear code of block length $n$ of rate $1 - \ent(\delta)$, and relative covering radius $\delta + \tilde{O}_{\delta}(1/\log n)$.}

{This ``gap to optimality'' is an improvement over (to the best of our knowledge) the previously best known construction, which only gave \linebreak $O(1/\sqrt{\log n})$.}

{In this paper, we will mainly be concerned with \emph{linear} covering codes, i.e., $\mathcal{C}$ is a subspace of $\F_2^n$, where one identifies $\F_2^n$ and $\{0,1\}^n$ in the obvious way. A code $\mathcal{C} \subseteq \F_2^n$, is said to have \emph{relative distance} $\delta$ if any two distinct $x,y \in \mathcal{C}$ are at least $\delta n$ apart in Hamming distance. The \emph{rate} of $\mathcal{C}$ is the quantity $\frac{\log_2|\mathcal{C}|}{n}$. As is well-known, the Gilbert--Varshamov (GV) \cite{gilbert-52-gv, varshamov-57-gv} bound says that a packing code of relative distance $\delta$ can achieve rate at least $1-\ent(\delta) - \epsilon$ for every $\epsilon > 0$ as $n$ grows. Here, $\ent(\cdot)$ denotes the binary entropy function. Surprisingly, such a naive bound is still more or less the best we know today. The best upper bound on packing rate is given by the Linear Programming (LP) bound \cite{mrrw-77-mrrw} which does not match GV bound in general.

 For covering, a very straightforward argument, called the \emph{sphere-covering bound} tells us that a code with relative covering radius $\gamma$ has to be of size at least $\frac{2^n}{|B_{\gamma n}|}$, where $ |B_{\gamma n}| $ denotes the volume of a Hamming ball of radius $\gamma n$. In coding theoretic terms (plus standard asymptotics), the \emph{rate} of a covering code with \emph{relative covering radius} $\gamma$ is at least $1-\ent(\gamma)$. By a simple application of the probabilistic method, one can show that a random code of rate $1 - \ent(\gamma) + \Omega(1/n)$ has relative covering radius at most $\gamma$ with high probability. A result of Blinovsky \cite{blinovsky-90-covering-linear} says that this is also true for random linear codes with probability at least $1 - o(1/n)$.}

{The aforementioned theorem of Blinovsky immediately gives a procedure to construct a covering code of relative covering radius $\delta$ and rate $1 - \ent(\delta) - \epsilon$ of block length $n$ in time $\operatorname{poly}(n)$ provided $n \geq \exp\left(O(1/\epsilon^2) \right)$. This is achieved simply by concatenating\footnote{By concatenation, we mean combining inner and outer codes in Forney's sense \cite{forney-65-concat}. If the outer code is not specified, then it is the identity code by default. In this case, concatenation can also be seen as taking direct sum of inner codes.} all possible linear codes. The details of this procedure are discussed formally in Section~\ref{sec:conc}. The improvement in this paper comes from concatenating a much smaller subset of linear codes, most of which attain the sphere-covering bound.}

\subsection{Notation and preliminaries}

\subsubsection{General}

{For a prime power $q$, we use $\F_q$ to denote the finite field of order $q$. As is common, we identify $\F_2^n$ with $\{0,1\}^n$ in the usual way, so for example, our codes are thought of as subsets of $\F_2^n$}
 
{Throughout this paper, we also use the fact that $\F_{2^n}\cong\F_2^n$ as an additive group. Where there is no ambiguity, we refer to an element $x \in \F_{2^n}$ to also mean an element (given by any fixed isomorphism) in $\F_2^n$. Since multiplication by an element $\alpha \in \F_{2^n}$ is a linear operation, it can be thought of as the action of a matrix $M_\alpha\in\F_2^{n\times n}$ on $\F_2^n$.}

{For a subset $S \subseteq \F_2^n$, we use $\langle S \rangle$ to mean the span of elements in $S$, i.e., $\langle S \rangle = \{\sum_{a \in T}a~|~T \subseteq S\}$. For subsets $A,B \subset H$ where $H$ is some group, we denote the sumset $A+B = \{a + b~|~ a \in A,~b \in B\}$. We let $B_r = B_r^n \subseteq \{0,1\}^n$ denote the Hamming ball of radius $r$ centered at $(0,\ldots,0)$. In terms of this notation, $\mathcal{C} \subseteq \F_2^n$ is a covering code with relative covering radius $\delta$ if $\mathcal{C} + B_{\delta n } = \F_2^n$.

\iffalse For a code $\mathcal{C}$, we define the \emph{covering radius} of $\mathcal{C}$, denoted by $r_{\operatorname{cov}}(\mathcal{C})$ as follows:}
\[
r_{\operatorname{cov}}(\mathcal{C}) : = \min\{r~|~\mathcal{C} + B_{r} = \F_2^n\}
\]
i.e., $\mathcal{C}$ is a covering code of relative covering radius $\delta$ only if $r_{\operatorname{cov}}(\mathcal{C}) \leq \delta n$.
\fi
{A code $\mathcal{C} \subseteq \F_2^n$ can also be thought of as a map from $\F_2^k \rightarrow \F_2^n$ where $k =\log_2 |\mathcal{C}| = Rn$, where $R$ is the rate of $\mathcal{C}$. This map is usually called the \emph{encoding map}, and where there is no ambiguity, we think of $\mathcal{C}$ as both a subset of $\F_2^n$ and an encoding map $\mathcal{C} : \F_2^k \rightarrow \F_2^n$. We call $\mathcal{C}$ \emph{linear} if it defines a linear map $x \rightarrow xG$ where $G \in \F_2^{k \times n}$ is called the \emph{generator matrix}.} 

{We will use the fact that for any $\delta \in (0,1)$, it holds that $\frac{1}{n}\log_2(|B_{\delta n}^n|) = \ent(\delta) \pm O\left( \frac{1}{n}\right)$. Finally, throughout the paper, all asymptotics/high probability events are for large/growing $n$.}

\subsubsection{The Wozencraft ensemble}

{The Wozencraft ensemble is the set of linear codes $\mathcal{W}^{(n)} := \{C_\alpha^n~|~\alpha\in\F_{2^n}\}$ of rate $1/2$ defined as follows:
$
C_\alpha^n:  \F_{2^n}  \to  \F_{2^n}\times \F_{2^n}
$
is the code that maps
$
x \mapsto  (x,\alpha x)
$.
{We are using the aforementioned identification of $\F_{2^{n}}$ and $\F_2^{n}$ in the definition above so $C_{\alpha}^n$ produces a code of block length $2n$. Since $\mathcal{W}^{(n)}$ is defined for every $n$, the Wozencraft ensemble is in fact, a \emph{family} of sets of linear codes. For $\alpha \in \F_{2^n}$, the map $C_{\alpha} : \F_2^n \rightarrow \F_2^{2n}$ is given by the generator matrix $C_{\alpha} = [I|M_\alpha ]\in\F_2^{n\times 2n}$.

The Wozencraft ensemble is very interesting and useful in its own right. Most Wozencraft codes $C_{\alpha}^n$ for $\alpha \in\F_{2^n}$ meet the Gilbert--Varshamov bound \cite{massey-63-wozencraft}. Moreover, they can be constructed in  $2^{O(n)}$ time. This makes them possible to use them as inner codes together with other algebraic codes (e.g., Reed--Solomon codes \cite{reed-solomon-60-rs}) as outer codes in various constructions of concatenated codes. This procedure gives rise to explicit codes with constant rate ($R > 0$) and constant distance ($\delta > 0$). Examples include Forney codes \cite{forney-65-concat}, which are polynomial-time constructible\footnote{A linear code is said to be \emph{polynomial-time constructible} if its generator matrix can be computed in poly$(n)$ time. This is what \emph{explicit} usually mean.} and Justesen codes \cite{justesen-72-justesen}, which are locally polynomially computable\footnote{We say that a linear code is \emph{locally polynomially computable} if each entry of the generator matrix can be computed in poly$(\log n)$ time. People also call such codes \emph{fully explicit}. It is a more stringent notion than polynomial-time constructability.}.
}

\subsubsection{Probability}

\paragraph{Chebyshev's Inequality} The first fact we will use is Chebyshev's Inequality which says for a random variable $X$,
$\Prob(|X - \E[X]| \geq t) \leq {\Var(X)}/{t^2}$.

We actually only need a straightforward corollary of this which states that for a nonnegative (discrete) random variable $X$, it is a corollary of Chebyshev Inequality that
\begin{equation}
\label{eqn:Chebyshev}
\Prob(X=0)\le\frac{\Var(X)}{\E^2[X]}.
\end{equation}

\paragraph{A martingale inequality} A sequence of random variables $X_0,X_1,\ldots, X_n$ is called a \emph{martingale} with respect to another sequence of random variables $Y_0, Y_1,\ldots, Y_n$ if for all $i \in [n-1]$, we have $X_i = f_i(Y_1,\ldots Y_i)$ for some function $f_i$, and $\E[X_{i+1}|Y_i,\ldots, Y_1] = X_i$.

We say that a martingale has the \emph{bounded difference property} (or, more precisely, the $1$-bounded difference property) if for every $i \geq 1$, $|X_i - X_{i-1}| \leq 1$. We have the Azuma-Hoeffding inequality for martingales (see~\cite{DP09}) with the bounded difference property which says
\begin{equation}
\label{eqn:azuma}
\Prob(X_n - X_0 \leq -t) \leq \exp\left(\frac{-t^2}{2n}\right).
\end{equation}

We will actually need the following straightforward corollary of this:

\begin{corollary}
\label{corr:martingale}
Suppose that $Y_1,Y_2,\ldots,Y_t$ are Bernoulli random variables such that for every $i \geq 1$, we have that
$\E[Y_i = 1 | Y_{i - 1},Y_{i - 2},\ldots] \geq \frac{1}{2}$.
Then we have 
\[
\Prob(\sum_{i = 1}^tY_i \leq t/4) \leq \exp\left(-\frac{t}{32}\right).
\]
\end{corollary}

\begin{proof}
Denote $Y := \sum_{i = 1}^t Y_i$ and define the sequence of random variables $X_0,\ldots, X_t$ where $X_i := \E[Y|Y_1,\ldots,Y_i]$. So we have $X_0 = \E[Y] \geq t/2$ and $X_t = Y$. Moreover, we have

\begin{align*}
\E[X_i |Y_{0},\ldots,Y_{i-1}] &  = \E[\E[Y|Y_{0},\ldots, Y_i]|Y_0,\ldots,Y_{i - 1}] \\
& = \E[Y|Y_0,\ldots, Y_{i-1}] \\
& =  X_{i-1}
\end{align*}

Thus the sequence $X_0,X_1,\ldots,X_t$ is a martingale with respect to $Y_0,Y_1,\ldots,Y_t$ with the bounded difference property, and the claim follows from (\ref{eqn:azuma}).
\end{proof}

\subsection{Our results and proof ideas}
 
Our main technical theorem shows the existence of a relatively few codes, most of which meet the GV bound for covering. More precisely, 

\begin{theorem}\label{thm:main}
There exists a constant $c>0$, $t=c\log n$, such that a linear code $\mathcal{C}$ of length $2n$ given by the generator matrix $G_0\in\F_2^{(n+t)\times (2n)}$ satisfies 
$\Prob(\mathcal{C}+A=\F_2^{2n})=1-o\left({1}/{n}\right)$, 
where the probability is taken over the random construction of $G_0$  defined as follows.
\[
G_0=\left[\begin{array}{c}
     C_{\alpha}  \\\cline{1-1}
     M 
\end{array}\right]
=\left[
\begin{array}{c}
    \begin{array}{c|c}
        I_n & M_\alpha 
    \end{array}\\\cline{1-1}
    M
\end{array}
\right],
\]
 where $\alpha $ in  $\F_{2^n}$ and $M\in\F_2^{t\times (2n)}$ are chosen uniformly at random, and $A$ is a Hamming ball of volume $n^3 \cdot 2^n$ centered at the origin.
In particular, the covering radius of $\mathcal C$ is about $ \ent^{-1}(1/2)\approx0.11n $.
\end{theorem}

{\paragraph{Remarks} Note that our code has rate $R = \frac{n+t}{2n} = \frac{1}{2} + O\left(\frac{\log n}{n}\right)$. Also recall that $A\subset\F_2^{2n}$ is taken to be a ball of volume $n^3\cdot2^n=2^{n+3\log n}$. In other words, $A$ has relative radius $r/n=\ent^{-1}\left(\frac{n+3\log n + O(1)}{2n}\right)  + O\left(\frac{\log n}{n}\right)= \ent^{-1}(1/2) + O\left(\frac{\log n}{n}\right)$. } 

{\paragraph{Proof idea} The proof of this Theorem~\ref{thm:main} is by first proving that for a randomly chosen code $C$ from the Wozencraft ensemble $\mathcal{W}^{(n)}$, it holds that with high probability, $|C + A| = 2^{2n}(1 - n^{-1})$. Thus, $C$ typically covers \emph{most} points in $\F_2^{2n}$ with relative radius $\ent^{-1}(1/2)$. This is proved by a modification of a straightforward second moment argument. The trick is essentially to consider $|C + A + \{b\}|$  (instead of $|C+A|$) for a uniformly chosen $b$. This provides us with some symmetry and extra randomness to carry out the second moment argument to show that for every point $u$, $|C \cap (A + \{u\})|$ is usually concentrated around its expectation and this, most points are covered by $C + A$.}

{Once we have the above claim, we then show that adding a few ($\approx \log n$) random translations of $|C+A|$ covers the whole of $\F_2^{2n}$. This just corresponds to adding the extra rows $M$.}

{This theorem immediately  gives the required construction via the discussion in Section~\ref{sec:conc}. A sketch is as follows:

{Covering radius behaves well w.r.t concatenation of codes, i.e., for codes $\mathcal{C}_1$ and $\mathcal{C}_2$ of covering radii $r_1$ and $r_2$ respectively, their concatenation has covering radius $r_1 + r_2$.}

{Now this immediately implies our construction by concatenating all codes from the support of $\mathcal{C}$ in Theorem~\ref{thm:main} of block length $k = \Theta\left(\frac{\log n}{\log \log n}\right)$ so that $n = k \cdot 2^k$. Theorem~\ref{thm:main} implies that at most a $\frac{1}{k}$ fraction of these codes have relative covering radius worse than $\ent^{-1}(1/2) + O\left(\frac{\log k}{k}\right)$. Thus the relative covering radius of the concatenated code is at most $\ent^{-1}(1/2) + O\left(\frac{\log k}{k}\right)$. Thus working backwards, suppose one wants a covering radius of $\ent^{-1}(1/2) + \epsilon$, the only constraint is that $\epsilon \geq \frac{\log k}{k}$ which is satisfied if $n \geq (1/\epsilon)^{O(1/\epsilon)}$.}

{Although we sketched the proof only for rate close to $\frac{1}{2}$, the construction works verbatim for any other rate $R \in (0,1)$. 
For rates different from $1/2$ we consider punctured codes.
The only difference would be to work with the Wozencraft ensemble $W_n$ restricted to the first $n+ k$ ($k<n$) coordinates. This gives codes of rate close to $\frac{n}{n + k}$, and the rest of the proof remains the same. Thus most of our attention will be focused on the case when the rate is close to $\frac{1}{2}$.}


\subsection{Related work}
Previously, explicit codes with low covering radius were constructed by Pach and Spencer \cite{pach-spencer-88-explicit-low-rcov}. When the covering radius is fixed, the asymptotic dependence on field size of  covering radius was investigated in \cite{davydov-giulietti-marcugini-oambianco-09-lin-nonbin-cov-codes}. Covering codes under different error models and with respect to (wrt) different metrics were also studied in the literature \cite{klove-schwartz-14-cov-codes-magnitude-errors, cooper-ellis-kahng-02-asymm-cov-codes}. Covering codes were also used in covert communication and steganography \cite{zhang-wang-zhang-07-cov-codes-steganography, bierbrauer-fridrich-08-cov-codes-steganography}. 

Beyond the scope of coding theory and information theory, covering codes also found their applications in cryptography \cite{guo-johansson-londahl-14-lpn-cov-codes}, complexity theory \cite{liu-18-chain-cov-codes-sat}, etc. Although we focus on combinatorial aspects of coverings, computational issues \cite{sloane-86-compute-rcov} arsing in  coverings also received significant attention. 

\section{Proof of Theorem~\ref{thm:main}}

First, we prove that a random element of the Wozencraft ensemble almost surely covers \emph{most} points of $\F_2^{2n}$.

\begin{lemma}\label{lem:stepi}
For $G := C_\alpha $ and $A$ as defined above, we have
$\Prob(|\F_2^{2n} \setminus (\langle G \rangle +  A)| \geq (1/n)2^{2n}) \leq {1}/{n^2}$.
\end{lemma}


\begin{proof}
Denote $W := \langle G \rangle$. Instead of working with $W$, we will look at $W' := \langle G \rangle + b$ for a randomly chosen $b \in \F_2^{2n}$ of the form $b = (\mathbf{0},b^{(2)})$, where $b^{(2)} \in \F_{2^n}$. 
It is easy to see that $|W + A| = |W' + A|$, and this shift provides us with some symmetry that will help in the analysis. For any $u\in\F_2^{2n}$, let us define $A_u := A + \{u\}$.

Consider any $a = (a^{(1)},a^{(2)}) \in \F_{2^{2n}}^2$. By definition of $W'$, we have that $a \in W'$ means there is some $x \in \F_{2^{2n}}$ such that $(a^{(1)},a^{(2)}) = (x,\alpha x) + (\mathbf{0}, b^{(2)})$. Therefore, we have  
\begin{align*}
    \Prob(a \in W') =& \sum_{\tau\in\F_{2^n}}\Prob(\alpha =\tau)\Prob(b^{(2)}=a^{(2)}-\tau a^{(1)})
    = \frac{1}{2^n},
\end{align*}
and so by linearity of expectation,
$ \E[|W' \cap A_u|] = |A_u| \cdot {2^{-n}} $.

We also have:
\begin{align*}
\E[|W' \cap A_u|^2] & = \sum_{a_1,a_2 \in A_u} \Prob(a_1 \in W' \land a_2 \in W').
\end{align*}

For distinct $a_1,a_2$, we observe that the event $\{a_1 \in W' \land a_2 \in W'\}$ holds if and only if there are distinct $x_1, x_2 \in \F$ such that 

\begin{enumerate}
\item\label{itm:cond1} $(a_1^{(1)},a_1^{(2)}) = (x_1,\alpha x_1) + (\mathbf{0},b^{(2)})$, and
\item\label{itm:cond2} $(a_2^{(1)},a_2^{(2)}) = (x_2,\alpha x_2) + (\mathbf{0},b^{(2)})$.
\end{enumerate}

Clearly, this gives us that $x_1 = a_1^{(1)}$ and $x_2 = a_2^{(1)}$. {Note that in order for both conditions \ref{itm:cond1} and \ref{itm:cond2} to hold, $b$ has to simultaneously satisfy the following equations:
$b^{(2)}=a_1^{(2)}-\alpha a_1^{(1)},\; b^{(2)}=a_2^{(2)}-\alpha a_2^{(1)}$.
For  pairs $(a_1^{(1)},a_2^{(1)})$ and $(a_1^{(2)},a_2^{(2)})$, if \emph{exactly one} pair among $\{a_1^{(1)}, a_1^{(2)}\}$ and $\{a_2^{(1)}, a_2^{(2)}\}$ are equal, then there is no feasible $b$. Otherwise such $b$ exists if $a_1^{(2)}-\alpha a_1^{(1)} = a_2^{(2)}-\alpha a_2^{(1)}$. }

Let $S : = \{(a_1,a_2) \in A^2~|~a_1 \neq a_2~\text{and}~a_1^{(1)} = a_2^{(1)}\}$. Recall that $r$ is the radius of the Hamming ball $A$. We bound
\begin{align}
|S|  = & {\sum_{j\le r}}\sum_{i \leq j}\binom{n}{i}\binom{n}{j - i}\left( \binom{n}{j - i} - 1\right) \notag\\
=&\sum_{j\le r}\sum_{i \leq j}\binom{n}{i}\binom{n}{j - i}\cdot o\left(\binom{n}{n/2}\right)\label{eqn:bound1}\\
\le &\sum_{j\le r}\sum_{i \leq j}\binom{n}{i}\binom{n}{j - i}\cdot o(2^n)\label{eqn:bound2}\\
=&o(2^n)\cdot\sum_{j\le r}\binom{2n}{j}\label{eqn:bound3}\\
 =& o(2^n \cdot |A_u|),\notag
\end{align}
where (\ref{eqn:bound1}) follows since the maximal value of $j-i$ is at most \linebreak $2n\cdot \ent^{-1}(1/2)\cdot(1+o(1))$ which itself is less than $2n\cdot(1/4)= n/2$, (\ref{eqn:bound2}) follows from using that $\binom{n}{n/2} = o(2^n)$, and (\ref{eqn:bound3}) follows from the Vandermonde convolution.


Therefore, back to the second moment calculation, we have:
\begin{align}
&\E[|W' \cap A_u|^2] \notag \\
& = \sum_{a_1,a_2 \in A_u} \Prob(a_1 \in W' \land a_2 \in W') \notag\\
& = \sum_a\Prob(a\in W') \notag \\
&+\sum_{a_1\ne a_2}\sum_{\tau\in\F_{2^n}}\Prob(\alpha =\tau)\Prob(a_1\in W'\land a_2\in W'|\alpha =\tau) \notag\\
& = \E[|W' \cap A_u|] \notag \\
&+ \sum_{\tau\in\F_{2^n}}\Prob(\alpha =\tau) \sum_{\substack{a_1, a_2\in A_u\\a_1^{(1)}\ne a_2^{(1)},a_1^{(2)}\ne a_2^{(2)}\\a_1^{(2)}-\tau a_1^{(1)} = a_2^{(2)}-\tau a_2^{(1)}}}\frac{1}{2^{2n}} \label{eqn:add_explain}\\
& \le |A_u|\cdot \frac{1}{2^n} + \left(|A_u|(|A_u| - 1) - |S|\right)\cdot \frac{1}{2^{2n}} \label{eqn:ineq}\\
& = \frac{|A_u|^2}{2^{2n}} + \frac{|A_u|}{2^n} - \frac{|A_u| + |S|}{2^{2n}},\notag
\end{align}
where Eqn. \ref{eqn:add_explain} follows by interchanging summations and noting that whenever $ a_1\ne a_2 $, $ \{a_1\in W'\} $ and $ \{a_2\in W'\} $ are independent and uniform due to the fact that $b$ was chosen uniformly. Inequality (\ref{eqn:ineq}) follows by dropping the last condition $a_1^{(2)}-\tau a_1^{(1)} = a_2^{(2)}-\tau a_2^{(1)}$.


Therefore, we have 
\begin{align*}
    \Var(|W' \cap A_u|) =& \E[|W' \cap A_u|^2] - \E^2[|W' \cap A_u]\\
    =& \frac{|A_u|}{2^n} - \frac{|A_u| + |S|}{2^{2n}}\\
    =& \frac{|A_u|}{2^n} - \frac{|A_u|+o(2^n|A_u|)}{2^{2n}}\\
    =& \frac{|A_u|}{2^n} - o\left(\frac{|A_u|}{2^n} \right)\\
    =& \frac{|A|}{2^n}(1-o(1)).
\end{align*}

Since the variance is small, inequality \ref{eqn:Chebyshev} gives us that:
\begin{equation}
\label{eqn:nocover}
\Prob(W' \cap A_u = \emptyset) \leq \frac{\Var(|W' \cap A_u|)}{\E^2[|W' \cap A_u|]} =  \frac{2^n}{|A|}(1 + o(1)).
\end{equation}


Let $X_u$ denote the indicator random variable for the event $\{W' \cap A_u = \emptyset\}$, i.e., that $u$ is not covered by $W' + A$. Denoting $X := \sum_{u}X_u$, the bound~(\ref{eqn:nocover}) gives us that $\E[X] \leq 2^{2n}\cdot \frac{2^n}{|A|}$, and so by Markov's inequality, we have that 
\[\Prob(X\geq n^2(2^n/|A|)2^{2n}) \leq \frac{1}{n^2},\]
which gives us the desired claim {by the choice of $|A|$}. 
\end{proof}

In the second phase, we argue that the union of $t$ random translations of the almost covering Wozencraft code obtained in the previous phase will cover the \emph{whole} space with high probability.

Let us call the uncovered points at the current stage $U = U_0:=\F_2^{2n}\setminus(W+A)$, and  let  $C = C_0 := W + A$ be the covered points. For a positive integer $i$, and  a random vector $u_i \in \F_2^{2n}$ independently chosen for each $i$, let $C_i := C_{i-1} \cup (C_{i-1} + u_i)$, and $U_i := \F_2^{2n} \setminus C_i$ denote the set of covered and uncovered points at stage $i$, respectively. Note that at any stage $i$, $C_i\sqcup U_i=\F_2^{2n}$. The following lemma completes the second step in the proof of Theorem~\ref{thm:main}


\begin{lemma}\label{lem:stepii}
There is some constant $k>0$ large enough, such that for $t \geq k \log n$, we have
$\Prob(U_t \neq \emptyset) \leq {1}/{n^2}$.
\end{lemma}

\begin{proof}
First, we observe that
$\E[|U_{i+1}|||U_i|] = {|U_i|^2}/{2^{2n}}$.

Indeed, for any $w \in U_i$, denote $X_w$ as the indicator random variable of the event $\{w \not \in U_{i+1}\}$. We have 
\begin{align*}
    \E[X_u=1||U_i|] = & \Prob(\exists v\in C_i,\;v+u_i=w||U_i|) 
    =  \frac{|C_i|}{2^{2n}}.
\end{align*}
 and linearity of expectation gives us the desired identity:
\begin{align*}
    \E[|U_{i+1}|||U_i|]=&|U_i|\Prob(u\in U_{i+1}||U_i|)\\
    =&|U_i|(1-|C_i|/2^{2n})\\
    =&|U_i|^2/2^{2n}.
\end{align*}

For $i \geq 0$, denote $Y_i$ to be the indicator random variable for the event $\{|U_{i+1}| \leq 2\cdot (|U_i|^2/2^{2n})\}$. Markov's inequality gives us that $\Prob(Y_i = 1) \geq \frac{1}{2}$ for any $U_i$. So, we have that for any $Y_0,\ldots.Y_{i - 1}$,
\[
\E[Y_i|Y_1,\ldots,Y_{i-1}] \geq \frac{1}{2}
\]
and so by Corollary~\ref{corr:martingale}, we have that 

\begin{align}
\Prob(Y \leq 2\log n)  = & \Prob\left(Y\le\left(1-\frac{\mu-2\log n}{\mu}\right)\mu \right)\notag\\
\le&\frac{1}{n^{\Omega(k)}},\label{eqn:choose_c}
\end{align}
where in the last inequality (\ref{eqn:choose_c}) we set $k$ large enough that this probability is less than $n^{-2}$.

It is left to observe that given $\sum_{i = 1}^t Y_i \ge 2\log n$, we have that $U_t = \emptyset$. Indeed, since
\begin{align*}
    |U_t| \le & \left(\frac{2}{2^{2n}}\right)^{1+2+2^2+\cdots+2^{2\log n-1}}|U_0|^{2^{2\log n}}\\
    =&\left(\frac{2}{2^{2n}}\right)^{2^{2\log n}-1}|U_0|^{2^{2\log n}}\\
    =&\left(\frac{2|U_0|}{2^{2n}}\right)^{2^{2\log n}}\frac{2^{2n}}{2}\\
    \le&\frac{1}{2}\cdot2^{-n^2\log\frac{n}{2}+2n}\\
    <&1.
\end{align*}
\end{proof}

Finally, to finish the proof, let
\begin{align*}
    E := & \{C+A=\F_2^{2n}\},\\
    E_1:=&\{|\langle G\rangle +A|>(1-1/n)\cdot2^{2n}\}.
\end{align*}
Overall we have that
\begin{align*}
    \Prob(E)\ge&\Prob(E_1)\Prob(E|E_1)\\
    \ge&(1-1/n^2)\cdot(1-1/n^2)\\
    =&1-O\left(\frac{1}{n^2}\right).
\end{align*}

\subsection{Covering radius of concatenated codes}
\label{sec:conc}

Here we show that covering codes behave well under concatenation.


For codes $C_1 \subseteq \{0,1\}^{n_1}$ of distance $d_1$ and rate $R_1$ and $C_2 \subseteq \{0,1\}^{n_2}$ of distance $d_2$ and rate $R_2$, their direct sum $C_1\oplus C_2$ is defined as
\[C_1\oplus C_2 = \{(x_1,x_2)\suchthat x_1\in C_1,\;x_2\in C_2\}.\]
Clearly,  $C_1\oplus C_2$ has blocklength $n_1+n_2$, rate $R_1+R_2$ and minimum distance $\min\{d_1,d_2\}$. It is not hard to see that the covering radius of a direct sum code is the sum of its components, i.e., $\rcov(C_1\oplus C_2)=\rcov(C_1)+\rcov(C_2)$. Indeed, since for every point every point $(y,z) \in \{0,1\}^{n_1 + n_2}$, $y$ is at most a distance $\rcov(C_1)$ from $C_1$ and $z$ is at most at distance $\rcov(C_2)$ from $C_2$.

If $C_1$ and $C_2$ are linear codes generated by matrices $G_1\in\F_2^{n_1R_1\times n_1}$ and $G_2\in\F_2^{n_2R_2\times n_2}$, then the direct sum has a block diagonal generator matrix $G$ of the following form:
\[G=\begin{bmatrix}
G_1 & \\
 & G_2
\end{bmatrix}.\]

Since we known that our construction is covering almost surely, the direct sum operation allows us to construct explicit covering codes by concatenating all Wozencraft-type codes. Let $n':=n+t$. We just put all matrices $\{G_i\}_{i=1}^N$ of the form $G_0$ defined in Theorem \ref{thm:main}   along the diagonal and get a matrix $G$ of size $n'\cdot N$ by $(2n)\cdot N$. These matrices generate the Wozencraft-type ensemble $\{C_i\}_i$ and there are $N:=2^n\cdot2^{t\cdot 2n}=2^{O(n\log n)}$ many such matrices in total. This operation  results in a code $C$ with generator matrix $G$ of blocklength $(2n)\cdot N$ and rate $\frac{1}{2} + O\left( \frac{\log n}{n}\right)$. The relative covering radius of the direct sum is at most
\begin{align*}
    \rcov\left(\bigoplus_{i=1}^N C_i\right)
    &\le(1-1/n^2)^2\cdot \ent^{-1}\left(\frac{n+3\log n + O(1)}{2n} \right) \notag \\
    &+ (1-(1-1/n^2)^2)\cdot 1\\
    &=\ent^{-1}(1/2) + O\left(\frac{\log n}{n}\right).
\end{align*}

It remains to observe that $\frac{\log n}{n} = \tilde{O}\left( \frac{1}{\log (2n \cdot N)}\right)$.

\paragraph{A note on other rates}
The above construction started off with an ensemble of codes $\mathfrak{C}$ of rate $\frac{1}{2}$ (Wozencraft).  However, the construction can be generalized for other rates in a standard way. The only thing that was used about the Wozencraft ensemble in the proof of Theorem~\ref{thm:main} was that $\mathfrak{C}$ was supported on $2^n$ codes, and had the following property. Fix a message $m$, and choose a random code $C \in \mathfrak{C}$, $C$ sends the message $m$ to $(m,x)$ where $x$ is a uniformly random element of $\F_2^n$. One can check that the proof works verbatim when $\mathfrak{C}$ is supported on $2^k$ codes, and a random $C \in \mathfrak{C}$ sends a message $m$ to $(m,x)$ where $x$ is uniform in $\F_2^k$. Therefore, one can restrict the Wozencraft ensemble to a set of coordinates (or \emph{puncture} it) to achieve different rates. For an $k \times n$ matrix $M$, for $S \subset [k]$ and $T \subset [n]$, we use $M[S,T]$ to denote the submatrix where the rows are indexed by $S$ and columns are indexed by $T$. For every generator matrix $C_{\alpha}=[I|M_\alpha]$ from the Wozencraft ensemble, denote
\[
G_k=
\left[
\begin{array}{c|c}
I_n[[k],[k]] & M_\alpha[[k],[n]] \\
\end{array}
\right].
\]

Given any message $m' \in \F_2^k$, one can check that a randomly chosen $G_k$ takes $m'$ to $(m',x')$ where $x'$ is a uniform point in the row span of $M_\alpha[[k],[n]]$. To see this, note that since $\alpha m$ is uniform in $\F_2^n$, take 
\[m=(m'(1),\cdots,m'(k),\underbrace{0,\cdots,0}_{n-k})\in\F_2^n.\]
The image of each such $m$  under $M_\alpha$ is exactly $m'\cdot M_\alpha[[k],[n]]$ which is equal to $x'$ and hence each $x'$ is equally likely to be output.
Therefore, by a similar proof as above, one can check that the code generated by a randomly chosen matrix $G_0$ given by:
\[
\left[
\begin{array}{c}
G_k \\
\hline
M
\end{array}
\right]
\]
is almost surely a good covering code of rate $\frac{k+t}{n+k}=\frac{nR+c\log n}{n+nR}\stackrel{n\to\infty}{\to}\frac{R}{1+R}$ if we denote $k=nR$. Similar arguments show that one can truncate $G$ as
\[
G^k=
\left[
\begin{array}{c|c}
I_n & M_\alpha[[n],[k]] \\
\end{array}
\right]
\]
to get a code of rate $\frac{n+t}{n+k}=\frac{n+c\log n}{n+nR}\stackrel{n\to\infty}{\to}\frac{1}{1+R}$.

\section{Open problems}
The obvious first open question is to construct a family of explicit covering codes of block length $n$, rate $1 - \ent(\delta)$ and covering radius $\delta + n^{- \Omega(1)}$.

We believe that in Theorem~\ref{thm:main}, the rows $M$ are just an artifact for the analysis. In particular, we believe that a random code from the Wozencraft ensemble has covering radius $\ent^{-1}(1/2) + o(1)$. 

These extra rows $M$ are a barrier to understanding the covering radius of several other distributions over codes, for example, low-density parity-check (LDPC) codes, etc..

A \emph{quasicyclic} code $C\le\F_2^{2n}$ is a linear code of rate $1/2$ spanned by the rows of a matrix of the form $G=[I|Q]$, where $I\in\F_2^{n\times n}$ and $Q\in\F_2^{n\times n}$ is a \emph{circulant} matrix
\[M=\begin{bmatrix}
-r_1-\\
-r_2-\\
\cdots\\
-r_n-
\end{bmatrix}.\]
For any $i\in[n-1]$, the $(i+1)$-th row  is a one-bit right-shift of the $i$-th row, i.e., $r_{i+1}=\sigma(r_{i})$ where $\sigma\in S_n$ is a permutation
\begin{align*}
    \sigma=\begin{pmatrix}
    1 & 2 & 3 & \cdots  & n\\
    n & 1 & 2 & \cdots  & n-1
    \end{pmatrix}.
\end{align*}
If we sample a row $r$ uniformly at random from $\F_2^n$ and construct a corresponding code $C$, then $C$ is known \cite{gaborit-zemor-08-improve-gv} to attain GV bound with high probability. Actually, it beats GV bound by some lower order factor and is the best asymptotic existence result in the constant relative minimum distance regime. However, we are unable to show its covering property. One challenge that this shares with the Wozencraft ensemble is that there is only a small amount of randomness in the construction. The whole matrix $G$ is completely determined once any row or column of $Q$ is sampled.

\bibliographystyle{alpha}
\bibliography{references} 

\end{document}